%% file: main.tex
\documentclass[sigconf]{acmart}
\usepackage[utf8]{inputenc}

\usepackage{tikz}

\usepackage{mdsymbol} 
\usepackage{graphicx}
\usepackage[linesnumbered,ruled, noend]{algorithm2e}
\usepackage{mathtools}
\newcommand\Mycomb[2][^n]{\prescript{#1\mkern-0.5mu}{}C_{#2}}
\newcommand{\baseline}{\textsc{BB-Base}\xspace}
\newcommand{\bucketing}{\textsc{BB-Bucket}\xspace}
\newcommand{\parbucketing}{\textsc{ParBB-Bucket}\xspace}
\newcommand{\dbpedia}{\textbf{\texttt{DbPedia}}\xspace}
\newcommand{\dbpedias}{\textbf{\texttt{DB}}\xspace}
\newcommand{\twitter}{\textbf{\texttt{Twitter}}\xspace}
\newcommand{\twitters}{\textbf{\texttt{TW}}\xspace}
\newcommand{\amazon}{\textbf{\texttt{Amazon}}\xspace}
\newcommand{\amazons}{\textbf{\texttt{AM}}\xspace}
\newcommand{\livejournal}{\textbf{\texttt{Live-Journal}}\xspace}
\newcommand{\livejournals}{\textbf{\texttt{LJ}}\xspace}
\newcommand{\bonanza}{\textbf{\texttt{Bonanza}}\xspace}
\newcommand{\bonanzas}{\textbf{\texttt{BO}}\xspace}
\newcommand{\senate}{\textbf{\texttt{Senate}}\xspace}
\newcommand{\senates}{\textbf{\texttt{SE}}\xspace}
\newcommand{\house}{\textbf{\texttt{House}}\xspace}
\newcommand{\houses}
{\textbf{\texttt{HO}}\xspace}
\newcommand{\tracker}{\textbf{\texttt{Tracker}}\xspace}
\newcommand{\trackers}{\textbf{\texttt{TR}}\xspace}
\newcommand{\orkut}{\textbf{\texttt{Orkut}}\xspace}
\newcommand{\orkuts}
{\textbf{\texttt{OR}}\xspace}

\usepackage{subcaption}
\usepackage{amsmath}
\usepackage{graphicx}

\newcommand{\remove}[1]{}
\newcommand{\quotes}[1]{``#1''}
\newcommand{\butterfly}{\mathbin{\rotatebox[origin=c]{90}{$\hourglass$}}}

\newcommand{\balancebutterfly}{\mathbin{\rotatebox[origin=c]{90}{$\hourglass$}^\star}}

\newtheorem{definition}{Definition}

\newtheorem{theorem}{Theorem}

\title{Balanced Butterfly Counting in Bipartite-Network}

\author{Apurba Das}
\email{apurba@hyderabad.bits-pilani.ac.in }
\affiliation{%
  \institution{Birla Institute of Technology and Science Pilani }
  \city{Hyderabad }
  \country{India.}
}

\author{Aman Abidi}
\email{aman.abidi@ntu.edu.sg}
\affiliation{%
  \institution{Nanyang Technological University}
   \country{Singapore}}

\author{Ajinkya	Shingane}
\email{ f20191375@hyderabad.bits-pilani.ac.in}
\affiliation{%
  \institution{Birla Institute of Technology and Science Pilani }
  \city{Hyderabad }
  \country{India.}
}

\author{Mekala	Kiran}
\email{ p20220017@hyderabad.bits-pilani.ac.in}
\affiliation{%
  \institution{Birla Institute of Technology and Science Pilani }
  \city{Hyderabad }
  \country{India.}
}

\begin{document}

\begin{abstract}
\input{Abstract}
\end{abstract}

\maketitle

\input{Introduction}
\input{Problem_Definition}
\input{Baseline_Approach}
\input{Our_Approach}

\input{Extension}
\input{Experiments}

\input{Related_Work}

\input{Conclusion_and_References.tex}

\end{document}

%% file: Abstract
Bipartite graphs offer a powerful framework for modeling complex relationships between two distinct types of vertices, incorporating probabilistic, temporal, and rating-based information. While the research community has extensively explored various types of bipartite relationships, there has been a notable gap in studying \textit{Signed Bipartite Graphs}, which capture liking/ disliking interactions in real-world networks such as customer-rating-product and senator-vote-bill. Balance butterflies, representing $2 \times 2$ bicliques, provide crucial insights into antagonistic groups, balance theory, and fraud detection by leveraging the signed information. However, such applications require counting balance butterflies which remains unexplored. In this paper, we propose a new problem: \textit{counting balance butterflies in a signed bipartite graph}.
To address this problem, we adopt state-of-the-art algorithms for butterfly counting, establishing a smart baseline that reduces the time complexity for solving our specific problem. We further introduce a novel bucket approach specifically designed to count balanced butterflies efficiently. We propose a parallelized version of the bucketing approach to enhance performance. Extensive experimental studies on nine real-world datasets demonstrate that our proposed bucket-based algorithm is up to \textbf{120x} faster over the baseline, and the parallel implementation of the bucket-based algorithm is up to \textbf{45x} faster over the single core execution. Moreover, a real-world case study showcases the practical application and relevance of counting balanced butterflies.

%% file: Introduction.tex
\section{Introduction}
Many real-world data, such as movies and their actors, products with the customers who like the products, research papers, and authors can be naturally represented using the bipartite network. 
A bipartite graph $G = (U, V, E)$ comprises two disjoint or independent vertex sets $U$ and $V$, and an edge set $E$ such that $E \subseteq U \times V$. To connect it to the real-world example, let there be an actor-movie bipartite network where actors are the nodes in one bipartition and the movies are the nodes in the other. There will be an edge between an actor and a movie if the actor performs in the corresponding movie. 
Similarly, there can be a user-product bipartite network where users and the products serve as nodes in the two bipartitions (users in one bipartition and the products in the other). There will be an edge between a user and a product if the user uses the product.

Mining \textit{bipartite cohesive subgraphs}, such as $k$-wing, maximal biclique, $k$-bitruss, balanced maximal biclique, etc., have attracted a lot of attention due to its wide range of applications across various domains such as community detection \cite{abidi2020pivot,das2018shared}, finding highly collaborative research groups \cite{abidi2022maximising}, and personalized recommendations \cite{abidi2022searching}. A rectangle \cite{wang2014rectangle}, butterfly, or $4$-cycle \cite{aksoy2017measuring} (i.e., a $2\times 2$ complete bipartite subgraph) is the smallest subgraph that forms the basic building block of many bipartite cohesive subgraphs. In terms of importance, it is equivalent to a triangle (the basic building block of cohesive subgraphs) in a unipartite graph and serves as the smallest cohesion unit in a bipartite graph.

Therefore to understand and analyze the properties, e.g., \textit{clustering coefficient}, \textit{community structure}, etc., in bipartite graphs, butterfly counting has emerged as an important problem. The butterfly counting problem has been extensively studied for undirected bipartite graphs \cite{wang2019vertex,sanei2018butterfly,wang2022accelerated,sanei2019fleet,shi2020parallel,xu2022efficient}. Recently, the problem has been extended to the weighted bipartite graphs as well \cite{zhou2021butterfly}. 
Further, the butterfly subgraph structure is adapted for the signed bipartite graphs, and consequently, \textit{balanced butterfly} is introduced in \cite{derr2019balance}. A balanced butterfly in a signed bipartite graph differs from a butterfly in an unsigned bipartite graph in the sense that a balanced butterfly has an even number of negative edges. It is equivalent to the signed balanced triangle in signed unipartite graphs for classical balance theory, a key signed social theory. Singed balanced butterflies are critical to the properties of signed bipartite graphs. \textbf{To the best of our knowledge, we are the first to introduce the problem of balanced butterfly counting and propose an efficient and scalable algorithm for counting balanced butterflies in a signed bipartite graph}.

The motivation to study the problem of counting balanced butterflies originates from several real-world applications requiring knowledge of balanced butterflies in a signed bipartite network. Some of the important applications are as follows:
\begin{itemize}
    
    \item \textbf{Predicting presence of antagonistic groups.} Finding balanced cohesive subgraphs using a balanced butterfly as a motif, i.e., maximal balanced signed bicliques \cite{sun2022maximal}, in a senator-bill voting bipartite graph, helps to obtain a group of two types of senators who vote for and against (antagonistic ones) the bills. Such groups are considered to be balanced and ensure that any bill voted by such groups has been discussed thoroughly, which reduces the risk of the bill favouring a specific community. Since the problem of finding maximal balanced signed bicliques is NP-Hard, an efficient polynomial counting algorithm for balanced butterflies can be exploited to predict the presence of such antagonistic groups.
    \item \textbf{Validating balance theory in bipartite graphs.} Balance theory \cite{cartwright1956structural,heider1946attitudes} has proven to provide vast improvements in various graph problems such as modeling, measuring, and mining tasks related to signed networks when utilized in the form of triangles in unipartite graphs. 
    Recently, Derr et al. \cite{derr2019balance}, exploited the similarity between balanced triangles and balanced butterflies and established that balanced butterflies are found significantly more often than unbalanced ones in signed bipartite networks, hence counting such substructures is crucial. 
\end{itemize}

It is not straightforward to apply a state-of-the-art butterfly counting algorithm for solving our problem of counting balanced butterflies in a signed bipartite network because the state-of-the-art butterfly counting algorithm counts the number of connections. Still, for counting balanced butterflies, we need to check for the sign of the connections. Thus, we modify the state-of-the-art algorithm to facilitate checking the edges' signs to decide whether a butterfly is balanced. However, the major drawback of explicitly checking each butterfly is that many butterflies might not be balanced, but we spend time checking for that. Following this, we develop a novel bucketing-based technique that minimizes the amount of wasteful computation to the maximum extent possible. To summarize, our contributions are as follows:

\begin{enumerate}
    \item We formally define the problem of counting balanced butterflies in a signed bipartite network.
    \item We develop a non-trivial baseline \baseline based on the state-of-the-art butterfly counting algorithm~\cite{wang2022accelerated}. In this algorithm, instead of simply counting the number of connections, we focus on the sign of each connection to decide whether a butterfly is balanced or not and accordingly update the count.
    \item Next, we present an exact efficient, balanced butterfly counting algorithm \bucketing based on a novel bucketing concept. In this approach, we group the wedges into two buckets to apply the counting technique for computing a balanced butterfly count instead of explicitly checking for the sign of each connection. This provides orders of magnitude speedup over the baseline solution \baseline. For example, on real-world signed bipartite network \senate, \bucketing is \textbf{120x} faster than the baseline solution \baseline. We theoretically prove the correctness of \bucketing and analyze the time and space complexity of the proposed solution.
    \item We develop a parallel algorithm \parbucketing based on the efficient \bucketing algorithm with a theoretical guarantee that \parbucketing is \textit{work-efficient} with a \textit{logarithmic span} which is desirable for a scalable parallel algorithm.
    \item We empirically evaluate our algorithms to show that (1) \bucketing outperforms \baseline on most of the signed bipartite graph where the connections between two partitions are dense, and (2) \parbucketing provides upto \textbf{45x} speedup in a $64$ core system compared to the runtime in a $1$ core system. Additionally, we show a case study from a movie dataset to show how we can analyze real-world datasets using balanced butterflies.
\end{enumerate}

\begin{figure}[t]
    \centering
    \scalebox{0.13}{
    \begin{tikzpicture}
        
    
    

    \end{tikzpicture}
    }
   \scalebox{0.22}{
   
    \begin{tikzpicture}
        \node(p)[] at (-1,9.5) { \Huge Butterflies};
            \draw [, dashdotdotted] (-2.6,-1) -- (36.6,-1);
            \draw [, dashdotdotted] (36.6,-1) -- (36.6,9);
            \draw [, dashdotdotted] (36.6,9) -- (-2.6,9);
            \draw [, dashdotdotted] (-2.6,9) -- (-2.6,-1);
            
            \node(p)[] at (7.5,8.5) { \Huge Balanced butterflies};
            \draw [blue, dashdotdotted] (-2,-.5) -- (22.6,-.5);
            \draw [blue, dashdotdotted] (22.6,-.5) -- (22.6,8);
            \draw [blue, dashdotdotted] (22.6,8) -- (-2,8);
            \draw [blue, dashdotdotted] (-2,8) -- (-2,-.5);
            
            \node(p)[] at (31,8.5) { \Huge Unbalanced butterflies};
            \draw [orange, dashdotdotted] (23,-.5) -- (36,-.5);
            \draw [orange, dashdotdotted] (36,-.5) -- (36,8);
            \draw [orange, dashdotdotted] (36,8) -- (23,8);
            \draw [orange, dashdotdotted] (23,8) -- (23,-.5);
    \end{tikzpicture}
    \hspace{-1100pt}
    \begin{subfigure}[b]{0.25\textwidth}
        \centering
        \begin{tikzpicture}
            \foreach \name/ \x in {v_1/0}
                \node (\name) at ( \x,6) [] {\includegraphics[scale=0.05]{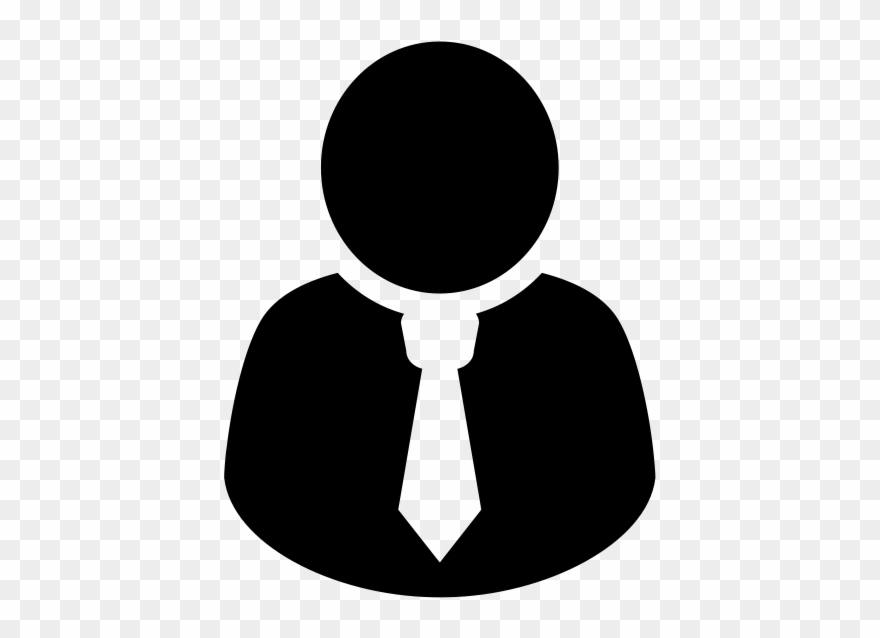} };
            \foreach \name/ \x in {v_1/0}
                \node () at (\x,7) [] {\Huge $v_i$};
            
            \foreach \name/ \x in {v_2/4}
                \node (\name) at ( \x,6) [] {\includegraphics[scale=0.05]{author_clipart.png} };
            \foreach \name/ \x in {v_2/4}
                \node () at (\x,7) [] {\Huge $v_j$};
            
            \foreach \name/ \x in {u_1/0}
                \node (\name) at ( \x,1.3) [] {\includegraphics[scale=0.05]{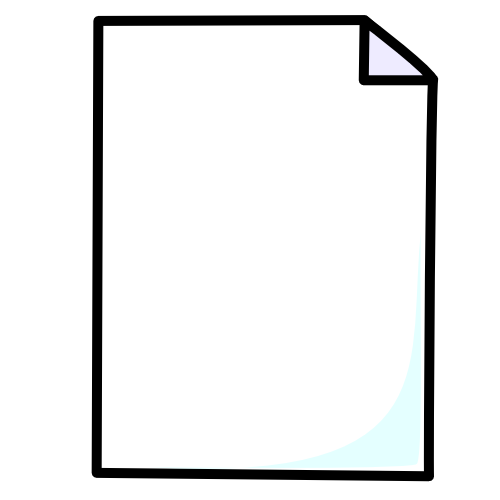} };
            \foreach \name/ \x in {u_1/0}
                \node () at (\x,.3) [] {\Huge $u_i$};
            
            \foreach \name/ \x in {u_2/4}
                \node (\name) at ( \x,1.3) [] {\includegraphics[scale=0.05]{paper_clip_art.png} };
            \foreach \name/ \x in {u_2/4}\node () at (\x,.3) [] {\Huge $u_j$};
        
            \node(p)[] at (-.5,4) {\Huge +};
            \node(p)[] at (1.45,5) {\Huge +};
            \node(p)[] at (2.55,5) {\Huge +};
            \node(p)[] at (4.5,4) {\Huge +};

            \path[]
            (v_1)
            edge  [very thick,green](u_1)
            edge  [very thick,green](u_2)
            
            (v_2) 
            edge  [very thick,green](u_1)
            edge  [very thick,green](u_2)
            ;
        \end{tikzpicture}
        
        \caption{\huge $ (++++)$}
        \label{fig:graphA}
    \end{subfigure}
    \hspace{1cm}
    \begin{subfigure}[b]{0.3\textwidth}
        \centering
        \begin{tikzpicture}
            \foreach \name/ \x in {v_1/0}
                \node (\name) at ( \x,6) [] {\includegraphics[scale=0.05]{author_clipart.png} };
            \foreach \name/ \x in {v_1/0}
                \node () at (\x,7) [] {\Huge $v_i$};
            
            \foreach \name/ \x in {v_2/4}
                \node (\name) at ( \x,6) [] {\includegraphics[scale=0.05]{author_clipart.png} };
            \foreach \name/ \x in {v_2/4}
                \node () at (\x,7) [] {\Huge $v_j$};
            
            \foreach \name/ \x in {u_1/0}
                \node (\name) at ( \x,1.3) [] {\includegraphics[scale=0.05]{paper_clip_art.png} };
            \foreach \name/ \x in {u_1/0}
                \node () at (\x,.3) [] {\Huge $u_i$};
            
            \foreach \name/ \x in {u_2/4}
                \node (\name) at ( \x,1.3) [] {\includegraphics[scale=0.05]{paper_clip_art.png} };
            \foreach \name/ \x in {u_2/4}\node () at (\x,.3) [] {\Huge $u_j$};
        
            \node(p)[] at (-.5,4) {\Huge +};
            \node(p)[] at (1.45,5) {\Huge +};
            \node(p)[] at (2.55,5) {\Huge -};
            \node(p)[] at (4.5,4) {\Huge -};
            
            \path[]
            (v_1)
            edge  [very thick,green](u_1)
            edge  [very thick,green](u_2)
            
            (v_2) 
            edge  [very thick,red](u_1)
            edge  [very thick,red](u_2)
            ;
        \end{tikzpicture}
        \caption{\huge $ (++--)$}
        \label{fig:graphB}
    \end{subfigure}
    \hspace{1cm}
    \begin{subfigure}[b]{0.3\textwidth}
        \centering
        \begin{tikzpicture}
            \foreach \name/ \x in {v_1/0}
                \node (\name) at ( \x,6) [] {\includegraphics[scale=0.05]{author_clipart.png} };
            \foreach \name/ \x in {v_1/0}
                \node () at (\x,7) [] {\Huge $v_i$};
            
            \foreach \name/ \x in {v_2/4}
                \node (\name) at ( \x,6) [] {\includegraphics[scale=0.05]{author_clipart.png} };
            \foreach \name/ \x in {v_2/4}
                \node () at (\x,7) [] {\Huge $v_j$};
            
            \foreach \name/ \x in {u_1/0}
                \node (\name) at ( \x,1.3) [] {\includegraphics[scale=0.05]{paper_clip_art.png} };
            \foreach \name/ \x in {u_1/0}
                \node () at (\x,.3) [] {\Huge $u_i$};
            
            \foreach \name/ \x in {u_2/4}
                \node (\name) at ( \x,1.3) [] {\includegraphics[scale=0.05]{paper_clip_art.png} };
            \foreach \name/ \x in {u_2/4}\node () at (\x,.3) [] {\Huge $u_j$};
        
            \node(p)[] at (-.5,4) {\Huge +};
            \node(p)[] at (1.45,5) {\Huge -};
            \node(p)[] at (2.55,5) {\Huge +};
            \node(p)[] at (4.5,4) {\Huge -};
            
            \path[]
            (v_1)
            edge  [very thick,green](u_1)
            edge  [very thick,red](u_2)
            
            (v_2) 
            edge  [very thick,green](u_1)
            edge  [very thick,red](u_2)
            ;
        \end{tikzpicture}
        \caption{\huge $ (+-+-)$}
        \label{fig:graphC}
    \end{subfigure}
    \hspace{1cm}
    \begin{subfigure}[b]{0.3\textwidth}
        \centering
        \begin{tikzpicture}
            \foreach \name/ \x in {v_1/0}
                \node (\name) at ( \x,6) [] {\includegraphics[scale=0.05]{author_clipart.png} };
            \foreach \name/ \x in {v_1/0}
                \node () at (\x,7) [] {\Huge $v_i$};
            
            \foreach \name/ \x in {v_2/4}
                \node (\name) at ( \x,6) [] {\includegraphics[scale=0.05]{author_clipart.png} };
            \foreach \name/ \x in {v_2/4}
                \node () at (\x,7) [] {\Huge $v_j$};
            
            \foreach \name/ \x in {u_1/0}
                \node (\name) at ( \x,1.3) [] {\includegraphics[scale=0.05]{paper_clip_art.png} };
            \foreach \name/ \x in {u_1/0}
                \node () at (\x,.3) [] {\Huge $u_i$};
            
            \foreach \name/ \x in {u_2/4}
                \node (\name) at ( \x,1.3) [] {\includegraphics[scale=0.05]{paper_clip_art.png} };
            \foreach \name/ \x in {u_2/4}\node () at (\x,.3) [] {\Huge $u_j$};
        
            \node(p)[] at (-.5,4) {\Huge -};
            \node(p)[] at (1.45,5) {\Huge -};
            \node(p)[] at (2.55,5) {\Huge -};
            \node(p)[] at (4.5,4) {\Huge -};
            
            \path[]
            (v_1)
            edge  [very thick,red](u_1)
            edge  [very thick,red](u_2)
            
            (v_2) 
            edge  [very thick,red](u_1)
            edge  [very thick,red](u_2)
            ;
        \end{tikzpicture}
        \caption{\huge $ (----)$}
        \label{fig:graphD}
    \end{subfigure}
    \hspace{1cm}
    \begin{subfigure}[b]{0.3\textwidth}
        \centering
        \begin{tikzpicture}
            \foreach \name/ \x in {v_1/0}
                \node (\name) at ( \x,6) [] {\includegraphics[scale=0.05]{author_clipart.png} };
            \foreach \name/ \x in {v_1/0}
                \node () at (\x,7) [] {\Huge $v_i$};
            
            \foreach \name/ \x in {v_2/4}
                \node (\name) at ( \x,6) [] {\includegraphics[scale=0.05]{author_clipart.png} };
            \foreach \name/ \x in {v_2/4}
                \node () at (\x,7) [] {\Huge $v_j$};
            
            \foreach \name/ \x in {u_1/0}
                \node (\name) at ( \x,1.3) [] {\includegraphics[scale=0.05]{paper_clip_art.png} };
            \foreach \name/ \x in {u_1/0}
                \node () at (\x,.3) [] {\Huge $u_i$};
            
            \foreach \name/ \x in {u_2/4}
                \node (\name) at ( \x,1.3) [] {\includegraphics[scale=0.05]{paper_clip_art.png} };
            \foreach \name/ \x in {u_2/4}\node () at (\x,.3) [] {\Huge $u_j$};
        
            \node(p)[] at (-.5,4) {\Huge +};
            \node(p)[] at (1.45,5) {\Huge -};
            \node(p)[] at (2.55,5) {\Huge +};
            \node(p)[] at (4.5,4) {\Huge +};
            
            \path[]
            (v_1)
            edge  [very thick,green](u_1)
            edge  [very thick,red](u_2)
            
            (v_2) 
            edge  [very thick,green](u_1)
            edge  [very thick,green](u_2)
            ;
        \end{tikzpicture}
        \caption{\huge $ (+-++)$}
        \label{fig:graphE}
    \end{subfigure}
    \hspace{1cm}
    \begin{subfigure}[b]{0.3\textwidth}
        \centering
        \begin{tikzpicture}
            \foreach \name/ \x in {v_1/0}
                \node (\name) at ( \x,6) [] {\includegraphics[scale=0.05]{author_clipart.png} };
            \foreach \name/ \x in {v_1/0}
                \node () at (\x,7) [] {\Huge $v_i$};
            
            \foreach \name/ \x in {v_2/4}
                \node (\name) at ( \x,6) [] {\includegraphics[scale=0.05]{author_clipart.png} };
            \foreach \name/ \x in {v_2/4}
                \node () at (\x,7) [] {\Huge $v_j$};
            
            \foreach \name/ \x in {u_1/0}
                \node (\name) at ( \x,1.3) [] {\includegraphics[scale=0.05]{paper_clip_art.png} };
            \foreach \name/ \x in {u_1/0}
                \node () at (\x,.3) [] {\Huge $u_i$};
            
            \foreach \name/ \x in {u_2/4}
                \node (\name) at ( \x,1.3) [] {\includegraphics[scale=0.05]{paper_clip_art.png} };
            \foreach \name/ \x in {u_2/4}\node () at (\x,.3) [] {\Huge $u_j$};
        
            \node(p)[] at (-.5,4) {\Huge -};
            \node(p)[] at (1.45,5) {\Huge +};
            \node(p)[] at (2.55,5) {\Huge -};
            \node(p)[] at (4.5,4) {\Huge -};
            
            \path[]
            (v_1)
            edge  [very thick,red](u_1)
            edge  [very thick,green](u_2)
            
            (v_2) 
            edge  [very thick,red](u_1)
            edge  [very thick,red](u_2)
            ;
        \end{tikzpicture}
        \caption{\huge $ (-+--)$}
        \label{fig:graphF}
    \end{subfigure}
    \hspace{1cm}
  }  
\caption{{\small{Possible singed balanced and unbalanced butterflies where green (red) edges denote positive (negative) relationship.}}}
\label{fig:my_label_motivation}
\vspace{-15pt}
\end{figure}

%% file: Problem_Definition.tex
\section{Problem Definition}
In this paper, we consider an undirected signed bipartite graph $G = (U, V, E)$, where $U$ and $V$ are the two independent vertex sets and $E \subseteq U \times V$ represents the edge set. For $\forall e(u, v) \in E$, a signed label of either \quotes{+} or \quotes{-} is associated with it. The label \quotes{+} denotes the positive (negative) relationship between the vertices, e.g., like (dislike), trust (distrust), etc. We now define some of the basic terminologies for our problem.

\begin{definition}\textbf{Butterfly $(\butterfly)$.}
    Given a bipartite graph $G=(U,V,E)$ and four vertices $u_i, u_j\in U, v_i,v_j \in V $, a butterfly induced by $u_i,v_i,u_j,v_j$ is a cycle of length of $4$ consisting of edges $(v_i,u_i)$, $(u_i,v_j)$, $(v_j,u_j)$ and $(u_j,v_i)$ $\in$ $E$. 
\label{def:butterfly}
\end{definition}

\begin{definition}
    \textbf{Balanced butterfly $(\balancebutterfly)$ \cite{derr2019balance}.} For a given signed bipartite graph $G$ and a butterfly $B$ is said to be balanced if it contains an even number of negative edges.
\label{def:bal_butterfly}
\end{definition}
Now, according to the Definition \ref{def:bal_butterfly}, there can be $6$ non-isomorphic butterflies as shown in Figure \ref{fig:my_label_motivation}. Among them, the first $4$ ($a$ to $d$) butterflies are balanced, i.e., even number of negative edges in a butterfly. 
The last $2$ ($e$ and $f$) butterflies are unbalanced butterflies in Figure \ref{fig:my_label_motivation}. 
Note, we will be only using red and green colours to represent negative and positive edges throughout the paper.

Further, we also familiarize the concept of a Wedge \cite{wang2019vertex}, which we would be utilizing later for our \textit{bucket approach}.
\begin{definition}
\textbf{Wedge $(\lor)$.} Given a signed bipartite graph $G(U, V, E)$ and vertices $u, v, w \in \{U \cup V\}$. A path starting from $u$, going through $v$ and ending at $w$ is called a wedge which is denoted as $\lor(u, v, w)$. 
For a wedge $\lor(u, v, w)$, we call $u$ the start-vertex, $v$ the middle-vertex and $w$ the end-vertex.
\label{def:wedge}
\end{definition}
Now, we are ready to define the problem statement as follows.

\vspace{5pt}
\noindent\textbf{Problem statement.} Given a signed bipartite graph $G=(U, V, E)$, the \textit{balanced butterfly counting problem} is defined as to determine the total number of balanced butterflies in $G$.

%% file: Baseline_Approach.tex
\section{Baseline Approach}\label{sec:baseline}
In this section, we propose the baseline algorithm for the balance butterfly counting problem. Since there does not exist any specific baseline solution to address our problem of counting balance butterfly, the naive approach for a baseline would be to enumerate all possible combinations of four vertices forming a butterfly and validate if the resulting butterfly is balanced. Although determining whether a butterfly is balanced or not requires constant time, exhaustive enumeration of all butterflies is very time-consuming, i.e., $O(|E|^4)$. Fortunately, there are a reasonable number of approaches to the problem of butterfly counting and its variations. Therefore, we propose a non-trivial baseline, where we modify the state-of-the-art butterfly counting algorithm to enumerate all butterflies and determine the eligible butterflies among them.
\subsection{Baseline Algorithm}
We now present our baseline that is inspired by vertex priority butterfly counting algorithm \cite{wang2019vertex}. The motivation for this approach was to prioritize all the vertices in both vertex sets, so as to reduce the overall number of wedges processed by the algorithm significantly. The correctness and theoretical reduction in running time complexity of the algorithm were also provided, which we discuss later. We start by defining the vertex priority for each vertex is defined as follows:
\begin{definition}
Vertex Priority \cite{wang2019vertex}: The vertex priority of any vertex $a \{\in U \cup V\} $ in comparison to any other vertex $b \in \{ U \cup V\} $ is defined as
$p(a)>p(b)$ if:
\begin{enumerate}
    \item $deg(a)$ $>$ $deg(b)$
    \item $id(a)>id(b)$, if $deg(a)$ $=$ $deg(b)$
\end{enumerate}
where $id(a)$ is the vertex ID of $a$.
\end{definition}
\begin{algorithm}[t]
\small
	\SetKwInOut{Input}{Input}
	\SetKwInOut{Output}{Output}
	\Input{$G$$(U,V,E)$, Input signed bipartite graph.}
	\Output{$\balancebutterfly$,  number of balanced butterflies.}
	\SetKwFunction{FMain}{Balance\_Butterfly}
    \SetKwProg{Fn}{Function}{}{}
    \Fn{\FMain{$G$}}{
        Calculate $p(u)$ for each $u \in \{U \cup V\}$\\
        Rearrange $\Gamma(u)$ for each $u \in \{U \cup V\}$ in ascending order\\
       $\balancebutterfly \gets 0$\\
       \ForEach{$u \in \{U \cup V\}$}{
       Initialise $H(w)$ for storing wedges corresponding to $v$\\
        \ForEach{$v \in \Gamma(u)$ $|$ $p(v)<p(u)$}{
        \ForEach{$w \in \Gamma(v)$ $|$ $p(w) < p(u)$}{
        $H(w).add(v)$
        }
        }
        \ForEach{$w$ $|$ $H(w)>1$}{
        \ForEach{vertex pair $(v_1,v_2)$ $\in H(w)$ $|$ $v_1 \neq v_2$}{
        \If{$\butterfly_{u,v_1,w,v_2}$.is\_Balanced}{
        $\balancebutterfly++$
        }
        }
        }
        }
  }
\caption{\textbf{\baseline}: Baseline approach}
\label{Algorithm:Baseline}
\end{algorithm}
Algorithm \ref{Algorithm:Baseline} describes the adaptation of the butterfly counting approach in \cite{wang2019vertex}. The key difference in our adaptation is that \cite{wang2019vertex} inputs all unsigned edges, which are indistinguishable and hence count all the possible butterflies, including the unbalanced edges. Moreover, \cite{wang2019vertex} treated all the butterflies equally (unsigned) hence saving the number of wedges was sufficient to calculate the total number of butterflies, i.e., for $k$ number of wedges with the same starting and ending vertices the total number of butterflies are given as $\Mycomb[k]{2}$. Therefore we modify the approach by storing the edges instead of just counting, which is then used to determine whether the butterfly is balanced or not. 

The algorithm starts by calculating the priority of all the vertices in $G$ (line 2). The sorting (ascending) of vertices is then performed based on the vertex priority (line 3). This allows to only processes the wedges where start-vertices have higher priorities than middle and end vertices, thus,  avoiding redundant wedges computation. The balanced butterfly is counted for each vertex $u \in \{U \cup V\}$ (lines 5-13). For each vertex $u$ a Hashmap $H$ is initialized to store all the corresponding wedges. Each corresponding wedge $(u,v,w)$ is enumerated such that $u$ and $w$ are $2$-hop neighbours and the $u$ has the highest priority. For each eligible wedge, the vertex $v$ is added to the list $H(w)$ in the Hashmap $H$. Consequently, for each vertex $u$, $H$ contains all the wedges. After that, for each combination of two vertices $v_1,v_2 \in H(w)$, corresponding to $u$, we check if the butterfly $\butterfly_{u,v_1,w,v_2}$ is balanced (lines 10-13). If so, we increment the count $\balancebutterfly$. Note that we can easily verify whether a butterfly is balanced or not within constant time.

\subsection{Complexity Analysis}
In this subsection, we examine the time and space complexities of our baseline algorithm.

\noindent\textbf{Time complexity.} We inspect the complexity by dividing the algorithm into two phases: (i) Initialization. (ii) Enumeration. For the initialization phase, the calculation of priority and sorting of the vertices can be performed in $O(|U|+|V|)$. The second phase of the enumeration of all balanced butterflies is the dominating part. In this phase, we start by computing all the wedges for an edge containing vertices $u$ and $v$, which requires $O(min\{def(u),deg(v)\})$. Due to the vertex priority order, we know that the start vertex of any wedge we discover must have a degree of size greater than or equal to the degree of the mid and end vertex. Therefore, the total number of wedges that can be uncovered by our algorithm containing a start vertex $u$ and an end vertex $w$ is $O(deg(w))$, which means the total combination of two wedges containing $u$ and $w$ is $deg(w)^2$. As a result, the cost of determining all wedge combinations of node $u$ and $w$ is $O(deg(w)^2$. Thus, the time complexity of our algorithm is $O(\sum_{e(u,v)\in E}min\{deg(u)^2,deg(v)^2\})$


\noindent\textbf{Space Complexity.} As discussed earlier the wedges enumerated by using the start vertex $u$ need to be stored. Therefore we can easily conclude the space required is the possible number of edges involving $u$, i.e., $O(deg(u))$.

%% file: Our_Approach.tex
\section{Bucket Approach}\label{sec:bucket}

In this section, we will be developing an efficient algorithm for counting balanced butterflies that minimizes enumerating unbalanced butterflies as much as possible. Our idea is based on first bucketing (grouping) the wedges in some way and then counting the balanced butterflies by combining the wedges from the buckets in a systematic way. We first propose some new terminologies utilized later in our approach.
\begin{definition}
\textbf{Symmetric Wedge $(\lor^s)$.} 
Given a wedge $\lor(u,v,w)$ in a bipartite graph $G = (U, V, E)$, a symmetric wedge is a wedge where both edges in the wedge, $(u,v)$ and $(v,w)$, have the same sign, either both positive or both negative and is denoted as $\lor^s(u,v,w)$.
\label{def:wedgesym}
\end{definition}
\begin{definition}
\textbf{Asymmetric Wedge $(\lor^{a})$.} 
Given a wedge $\lor(u,v,w)$ in a bipartite graph $G = (U, V, E)$, an asymmetric wedge is a wedge where both edges in the wedge, $(u,v)$ and $(v,w)$, have the opposite sign and is denoted as $\lor^{a}(u,v,w)$.
\label{def:wedgeasym}
\end{definition}

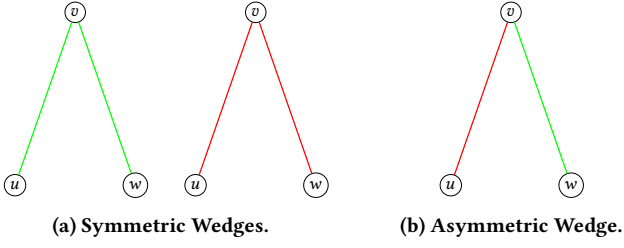
\begin{figure}[t]
  \centering
  \begin{minipage}{0.5\linewidth}
    \centering
    \scalebox{0.4}{
    \begin{tikzpicture}
      \node [circle,draw](v_1) at (2,7) {\Huge $v$ };
      \node [circle,draw](u_1) at (0,1.25) {\Huge $u$ };
      \node [circle,draw](u_2) at (4,1.25) {\Huge $w$ };
      \node [circle,draw](v_2) at (8,7) {\Huge $v$ };
      \node [circle,draw](u_3) at (6,1.25) {\Huge $u$ };
      \node [circle,draw](u_4) at (10,1.25) {\Huge $w$ };
      \path[]
        (v_1) edge [very thick,green] (u_1)
              edge [very thick,green] (u_2)
        (v_2) edge [very thick,red] (u_3)
              edge [very thick,red] (u_4);
    \end{tikzpicture}
    }
    \subcaption{Symmetric Wedges.}\label{subfig:sywed}
  \end{minipage}
  \hfill
  \begin{minipage}{0.4\linewidth}
    \centering
    \scalebox{0.4}{
    \begin{tikzpicture}
      \node [circle,draw](v_1) at (14,7) {\Huge $v$ };
      \node [circle,draw](u_1) at (12,1.25) {\Huge $u$ };
      \node [circle,draw](u_2) at (16,1.25) {\Huge $w$ };
      \path[]
        (v_1) edge [very thick,red] (u_1)
              edge [very thick,green] (u_2);
    \end{tikzpicture}
    }
    \subcaption{Asymmetric Wedge.}\label{subfig:asywed}
  \end{minipage}
  \caption{Unique types of wedges.}
  \label{fig:types_of_wedges}
\end{figure}
Fig. \ref{subfig:sywed} and \ref{subfig:asywed} represent symmetric and asymmetric wedges, respectively. The motivation behind introducing these wedges in a bipartite graph is to generate all types of balanced and unbalanced butterflies in  Fig. \ref{fig:my_label_motivation}. We further propose a lemma to establish the significance of these wedges.
\begin{lemma}
Any balanced butterfly in a signed bipartite graph can be constructed using a combination of symmetric and asymmetric wedges.
\label{lemma:wedges}
\end{lemma}
\begin{proof}
Let $ G = (U, V, E)$ be a signed bipartite graph, and let $\butterfly_{u,v,w,x}$ be a balanced butterfly in $G$, where $u,w \in U$ and $v,x \in V$. We will prove that $\butterfly_{u,v,w,x}$ can be constructed using a combination of symmetric and asymmetric wedges.

A butterfly is composed of two wedges, with the following case scenarios: Case $1$. Let the two edges $(u,v)$ and $(v,w)$ be of the same sign; then they form a symmetric wedge $\lor^s(u,v,w)$. Now for $\butterfly_{u,v,w,x}$ to be balanced needs the wedge $\lor^s(v,w,x)$ to be symmetric, i.e., edges $(w,x)$ and $(x,u)$ should be of the same sign. 

Case $2$. Assume that $(u,v)$ is positive and $(v,w)$ is negative, then they form an asymmetric wedge $\lor^a(u,v,w)$. Thus we can construct a balanced butterfly $\butterfly_{u,v,w,x}$ only if the wedge $\lor^s(v,w,x)$ is asymmetric, i.e., edges $(w,x)$ and $(x,u)$ should be of the opposite sign. 

Therefore, any butterfly constructed using a pair of symmetrical or asymmetrical wedges will always result in a balanced butterfly.
\end{proof} 
Further, we propose another lemma that forms the crux of our bucket approach.
\begin{lemma}
Given a signed bipartite graph $G=(U,V,E)$, for any pair of vertex $a \in U \cup V$ and $c$ (its $2$-hop neighbour), the total number of balanced butterflies containing $a$ and $c$ is given as $\Mycomb[l]{2}+\Mycomb[m]{2}$, where $l$ and $m$ are the numbers of symmetric and asymmetric wedges respectively containing $a$ and $c$. 
\label{lemma:total_balanced_butterflies}
\end{lemma}
\begin{proof}
Let $G$ be a signed bipartite graph, and let $\butterfly_{a,b,c,d}$ be a balanced butterfly in $G$. Using Lemma \ref{lemma:wedges}, we observe that for $\butterfly_{a,b,c,d}$ to be balanced, it must select two wedges of the same type, i.e., either two symmetric or two asymmetric wedges (Lemma \ref{lemma:wedges}). Therefore, we can select a pair of eligible symmetric or asymmetric wedges containing vertices $a$ and $c$ to form a balanced butterfly.

Let the number of symmetric and asymmetric wedges containing $a$ and $c$ be $l$ and $m$ respectively. The corresponding number of balanced butterflies is computed as the number of possible combinations of symmetric wedges and the number of possible combinations of asymmetric wedges.

The number of possible combinations of symmetric and asymmetric wedges containing vertices $a$ and $c$ is given by $\Mycomb[l]{2}$ and $\Mycomb[m]{2}$ respectively. Therefore, the total number of balanced butterflies containing vertices $a$ and $c$ is given as $\Mycomb[l]{2}+\Mycomb[m]{2}$.


\end{proof}

\subsection{Algorithm}
Suppose we want to generate a balanced $\butterfly_{u,v_1,w,v_2}$ where the priority of the vertex $u$ is the highest. For generating this butterfly, we will be using the ordered wedges from the buckets in some specific order.
\begin{algorithm}[t]
\small
	\SetKwInOut{Input}{Input}
	\SetKwInOut{Output}{Output}
	\Input{$G$$(U,V,E)$, Input signed bipartite graph.}
	\Output{$\balancebutterfly$,  number of balanced butterflies.}
	\SetKwFunction{FMain}{Balance\_Butterfly}
    \SetKwProg{Fn}{Function}{}{}
    \Fn{\FMain{$G$}}{
        Calculate $p(u)$ for each $u \in \{U \cup V\}$\\
        Rearrange $\Gamma(u)$ for each $u \in \{U \cup V\}$ in ascending order\\
       $\balancebutterfly \gets 0$\\
       $B_1\gets 0$\;
       $B_2\gets 0$\;
       \ForEach{$u \in \{U \cup V\}$}{
       Initialise $H(w)$ for storing wedges corresponding to $v$\\
        \ForEach{$v \in \Gamma(u)$ $|$ $p(v)<p(u)$}{
        \ForEach{$w \in \Gamma(v)$ $|$ $p(w) < p(u)$}{
        \If{$\sigma(u,v) == \sigma(v,w)$}{
        $B_1[w]++$\;
        }
        \Else{
        $B_2[w]++$\;
        }
        }
        }
        \ForEach{$w\in B_1$}{
            $\balancebutterfly \gets\balancebutterfly + \binom{B_1[w]}{2}$    
        }
        \ForEach{$w\in B_2$}{
            $\balancebutterfly \gets\balancebutterfly + \binom{B_2[w]}{2}$    
        }
        }
  }
\caption{\textbf{\bucketing}: Bucket Approach}
\label{Algorithm:Bucket}
\end{algorithm}
\begin{theorem}
Algorithm \ref{Algorithm:Bucket} correctly counts all the balanced butterflies.
\end{theorem}
\begin{proof}
The main purpose of Algorithm \ref{Algorithm:Bucket} is to eliminate the need for enumerating unbalanced butterflies in lines $10$ to $13$. From Lemma \ref{lemma:wedges}, we conclude that all the combinations of one symmetric and one asymmetric wedge are an unbalanced butterfly and, therefore, redundant. Furthermore, Lemma \ref{lemma:total_balanced_butterflies} allows us to calculate the total number of balanced butterflies per each pair of vertex $u$ and $w$ ($2$-hop neighbor). Also, ensuring each pair of a vertex, $2$-hop neighbors, is traversed only once by exploiting the vertex priority implies all the balanced butterflies are only counted once. Therefore for each $2$-hop neighbor vertex pair $u,w \in {U \cup V}$, all the wedges are enumerated and unlike Algorithm \ref{Algorithm:Baseline}, Algorithm \ref{Algorithm:Bucket} keeps track of the total number of wedges separately in two distinct buckets, i.e., $B_1$ for symmetric wedges and $B_2$ for asymmetric wedges, and utilizes Lemma \ref{lemma:total_balanced_butterflies} in lines $15$ to $18$ to accurately calculate $\balancebutterfly$. 
\end{proof}

\subsection{Complexity Analysis}

\noindent\textbf{Time Complexity:}~To analyze the time complexity, we observe that The algorithm works in two stages. The first stage is the preprocessing stage, where we compute the priorities of the vertices and sort the neighbors of each vertex in the adjacency list. Note that the sorting of the vertices takes $O(n)$ time using the binning technique~\cite{k-core-single-pc}, and the priority computation based on degree takes $O(1)$ time. The second stage is the counting stage, where we count the balanced butterflies by iterating over each vertex in the graph. Consider iteration over any vertex $u$ and its neighbor $v$. For this edge $(u,v)$, we work over the neighbors of $v$ only if $p(v) < p(u)$. Because we use only degree-based priority computation, the cost of processing an edge $(u,v)$ takes $O(\mbox{min}\{deg(u), deg(v)\})$. So the overall time complexity of the algorithm is $O(\Sigma_{e = (u,v)\in E}\mbox{min}\{deg(u), deg(v)\})$.\\

\noindent\textbf{Space Complexity:}~Apart from storing the graph, we use two additional data structures $B_1$ an $B_2$ for the maintenance of symmetric and asymmetric wedges. Now, the total number of entries in $B_1$ and $B_2$ can be at most $|U|$ if $u$ is from $U$, or it can be at most $|V|$ if $u$ is from $V$. Thus, the overall space complexity is $O(\mbox{max}\{|U|, |V|\})$.

%% file: Extension.tex
\section{Extensions}

In this section, we will extend our algorithm to compute (1) $\balancebutterfly_u$: balanced butterfly per vertex $u$; and (2) parallel implementation of Algorithm~\ref{Algorithm:Bucket}.  

\subsection{Counting for each vertex}

We modify the Algorithm~\ref{Algorithm:Bucket} for counting the number of balanced butterflies per vertex. Given a vertex $u$, we first look into its neighborhood $\Gamma(u)$, and for every vertex $v$ in $u$'s neighborhood, we will be looking into $v$'s neighbors $w\in\Gamma(v)$. For each wedge $<u,v,w>$, we will check the sign $\sigma(u,v)$ and $\sigma(v,w)$ to decide on which bucket we need to select to put the count of this wedge in. Finally, we count the number of balanced butterflies containing $u$ by combining the counts from each bucket. Note that we drop the rank comparison in choosing the wedges for counting all balanced butterflies containing a particular vertex. The pseudocode of this procedure is presented in Algorithm~\ref{vbbfc}.

\begin{algorithm}[ht!]
\small
	\SetKwInOut{Input}{Input}
	\SetKwInOut{Output}{Output}
	\Input{$G$$(U,V,E)$, Input signed bipartite graph, and a vertex $u$.}
	\Output{$\balancebutterfly_u$,  number of balanced butterflies containing vertex $u$.}
    $\balancebutterfly_u\gets 0$\;
    \For{$v\in\Gamma(u)$}{
        \For{$w\in\Gamma(v)$}{
            \If{$\sigma(u,v) == \sigma(v,w)$}{
                $H_1[w]++$\;
            }
            \Else{
                $H_2[w]++$\;
            }
        }
    }
    \For{$w\in H_1$}{
        $\balancebutterfly_u\gets\balancebutterfly_u+\binom{H_1[w]}{2}$\;
    }
    \For{$w\in H_2$}{
        $\balancebutterfly_u\gets\balancebutterfly_u+\binom{H_2[w]}{2}$\;
    }
 \caption{vBBFC($v, G$): Per vertex balanced butterfly counting}
 \label{vbbfc}
\end{algorithm}

\subsection{Parallelization}

One appealing aspect of our algorithm is that each vertex $u\in U\cup V$ in the graph $G = (U,V, E)$ can be processed independently. The only obstacle is adding the appropriate wedge information in one of the buckets in parallel. For this, we use a \textbf{parallel hash table} that supports insertion, deletion, and membership queries, and it performs $n$ operations in $O(n)$ \textbf{works} and $O(\log{n})$ \textbf{span} with high probability~\cite{FOCS-91}. The \textbf{work} $W$ of an algorithm is the total number of operations, and the \textbf{span} $S$ of an algorithm is the longest dependency path~\cite{parallel-algorithm}.

For each vertex $u$ in the bipartite graph, we look into each of its neighbors and their neighbors in parallel as there is no dependency in this lookup. Whenever we complete a valid wedge, we increment the count of that wedge starting at $u$ and ending at $w$. For a vertex $u$, we update the count of the wedges for each $w$ in parallel using the parallel hash tables $B_1$ and $B_2$. Finally, for each pair $(u, w)$, we count all butterflies containing both $u$ and $w$ by performing aggregation in parallel using atomic variable with an assumption that each atomic operation is performed in $O(1)$ time. The pseudocode is presented in Algorithm~\ref{Algorithm:parallel}.

\begin{algorithm}[t!]
\small
	\SetKwInOut{Input}{Input}
	\SetKwInOut{Output}{Output}
	\Input{$G$$(U,V,E)$, Input signed bipartite graph.}
	\Output{$\balancebutterfly$,  number of balanced butterflies.}
	\SetKwFunction{FMain}{Balance\_Butterfly}
    \SetKwProg{Fn}{Function}{}{}
    \Fn{\FMain{$G$}}{
        Calculate $p(u)$ for each $u \in \{U \cup V\}$\\
        Rearrange $\Gamma(u)$ for each $u \in \{U \cup V\}$ in ascending order\\
       $\balancebutterfly \gets 0$\\
       $B_1\gets 0$\;
       $B_2\gets 0$\;
       \ForPar{$u \in \{U \cup V\}$}{
       Initialise $H(w)$ for storing wedges corresponding to $v$\\
        \ForPar{$v \in \Gamma(u)$ $|$ $p(v)<p(u)$}{
        \ForPar{$w \in \Gamma(v)$ $|$ $p(w) < p(u)$}{
        \If{$\sigma(u,v) == \sigma(v,w)$}{
        $B_1[w]++$\;
        }
        \Else{
        $B_2[w]++$\;
        }
        }
        }
        \ForPar{$w\in B_1$}{
            $\balancebutterfly \gets\balancebutterfly + \binom{B_1[w]}{2}$    
        }
        \ForPar{$w\in B_2$}{
            $\balancebutterfly \gets\balancebutterfly + \binom{B_2[w]}{2}$    
        }
        }
  }
\caption{\parbucketing: parallelization of \bucketing}
\label{Algorithm:parallel}
\end{algorithm}

\noindent\textbf{Complexity Analysis:}~The total number of operations in the parallel algorithm (Algorithm~\ref{Algorithm:parallel}) is exactly the same as the total number of operations in the sequential algorithm (Algorithm~\ref{Algorithm:Bucket}). Therefore, the parallel algorithm is \textit{work efficient}.

For the analysis of the span of the parallel algorithm, note that, there might be at most $O(n^2)$ hash table operations in $B_1$ and $B_2$ (Line-$12$ and Line-$14$) and $O(n)$ additions (Lines-$15$ to Line-$18$) which contribute to $O(\log{n})$ span of the parallel algorithm \textbf{whp}. The high probability term comes from the parallel hash table operation.

%% file: Experiments.tex
\section{Experiments}\label{sec:experiment}
\subsection{Experimental Setup}

\noindent\textbf{Algorithms:}~We have implemented the following algorithms for the empirical evaluations.

\begin{itemize}
    \item \textbf{\baseline:}~The baseline algorithm for counting the number of balanced butterflies is described in Section~\ref{sec:baseline}.
    \item\textbf{\bucketing:}~This algorithm uses the bucketing concept to efficiently count balanced butterflies presented in Section~\ref{sec:bucket}.
    \item\textbf{\parbucketing:}~A parallel implementation of \bucketing. We use Intel TBB~\footnote{\url{https://github.com/oneapi-src/oneTBB}} library for the parallel implementation. In our implementation, we have used {\tt tbb::parallel\_for} for parallelizing the loop, {\tt tbb::concurrent\_hash\_map} for parallel operations on the buckets, and {\tt tbb::atomic} for atomic operations such as increment the butterfly count.
\end{itemize}

\remove{
We are reporting the results of our empirical studies that have been performed to evaluate the performance of different algorithms. Specifically, we have conducted experiments to compare the effectiveness and efficiency of the algorithms under consideration against one another. The results of our experiments are presented in this report, and they provide insights into the performance characteristics of each algorithm, allowing for informed decision-making about which algorithm to use in practice.
}

\noindent\textbf{Computing Resources.}~We implemented all the Algorithms in C++. We evaluated all the sequential implementations on Intel Xeon Silver 4114 (10 physical cores) with $2.2$GHz frequency and parallel implementations in our HPC server equipped with $64$ physical cores AMD EPYC 7542 processor and $256$GB main memory.


\remove{
We conducted all the experiments with C++ algorithms on a Linux server with HP ML10 Intel Xeon E3-1225 V5 processors and 32GB of main memory. If an algorithm runs for more than 10 hours, it is considered to have exceeded its maximum allowed running time, and it will be terminated. In other words, the algorithm is subject to a time limit of 10 hours, and if it has not completed its execution within this time frame, it will be stopped or terminated.
}

\subsection{Datasets}

We use $9$ datasets for our experiments out of which $3$ (\bonanza, \senate, and \house)~\footnote{\url{https://github.com/tylersnetwork/signed_bipartite_networks}} are the real-world signed bipartite graphs. We use rest of the datasets from Konect~\footnote{\url{http://konect.cc/networks/}}. Next, we create a signed bipartite graph in the following way: For each edge in the graph, we assign $1$ with probability $0.5$ to represent $+$ive edge and assign $0$ with probability $0.5$ to represent $-$ive edge. This completes our network construction. The description of the datasets is presented in Table~\ref{tab-1}.

\label{label:expsetting}
\noindent
\begin{table}[t]
\small
\caption{\label{table:datasets}
Characteristics of the Data sets}
\vspace*{-4mm}
\begin{center}
\scalebox{0.888}{
\begin{tabular}{|c|c||c|c||c|}
\hline
\textbf{Dataset}& $|E|$  & $|U|$ & $|V|$ &  $|\balancebutterfly|$  \\
\hline
\bonanza (\bonanzas) & $36$K & $7919$ & 
$1973$ & $641$K\\
\hline
\senate (\senates) & $27$K & $145$ & $1056$ & $15.3$M  \\
\hline
\house (\houses) & $114$K & $515$ & $1281$ & $280.8$M \\
\hline
\hline
\dbpedia (\dbpedias) & $293$K & $172$K & $53$K & $1.88$M \\
\hline
\twitter (\twitters) & $1.8$M & $175$K & $530$K & $206.5$M \\
\hline
\amazon (\amazons) & $5.8$M & $2.1$M & $1.2$M & $82.6$K\\
\hline
\livejournal (\livejournals)  & $112$M & $3.2$M & $7.4$M & $1.2$T \\
\hline
\tracker (\trackers)  & $140$M & $27.7$M & $12.7$M & $737.4$B \\
\hline
\orkut (\orkuts) & $327$M & $2.7$M & $8.7$M & $11$T \\
\hline
\end{tabular}}
\end{center}
\label{tab-1}
\end{table}

\remove{
\begin{table*}[ht]
\small
\caption{\label{table:datasets}
Results of Data sets}
\vspace*{-4mm}
\begin{center}
\scalebox{0.888}{
\begin{tabular}{|c||c|c||c|}
\hline
\textbf{Data set}& $|E|$  & $|U|$ & $|V|$ & $|\butterfly|$ &  $|\balancebutterfly|$  \\
\hline
\bonanza (\bonanzas) & $36$K & $9892 (7919 + 1973)$ & 
$671893$ & $641108$\\
\hline
\senate (\senates) & $27$K & $1201 (145 + 1056)$ & $25666956$ & $15323136$ \\
\hline
\house (\houses) & $114$K & $1796 (515 + 1281)$ & $469609963$ & $280793031$ \\
\hline
\hline
\dbpedia (\dbpedias) & $293$K & $225486(172K + 53K)$ & $3.76 \times 10^{6}$ & $1.88 \times 10^{6}$ \\
\hline
\twitter (\twitters) & $1.8$M & $705632(175K + 530K)$ & $2.07 \times 10^{-8}$ & $206508691$ \\
\hline
\amazon (\amazons) & $5.8$M & $3376972(2.1M + 1.2M)$ & $3.67 \times 10^{7}$ & $82569$\\
\hline
\livejournal (\livejournals)  & $112$M & $10690276 (3.2M + 7.4M)$ & $2021553152615$ & $1240565902932$ \\
\hline
\tracker (\trackers)  & $140$M & $40421974 (27665730,12756244)$ & $914160553338$ & $737385019678$ \\
\hline
\orkut (\orkuts) & $327$M & $11.5$M(2.7M + 8.7M) & $22.1$T & $11$T \\
\hline
\end{tabular}}
\end{center}
\end{table*}
}

\begin{table}[b]
\small
\caption{\label{table:datasets}
Runtime (sec.) on real-world data sets.}
\vspace*{-4mm}
\begin{center}
\scalebox{0.888}{
\begin{tabular}{|c|c|c|}
\hline
\textbf{Dataset}& \baseline  & \bucketing  \\
\hline
 \dbpedias & $1.271886$ & $0.231298$ (\textbf{5x})  \\
\hline
 \twitters & $95.548655$ & $18.247842$ (\textbf{5x}) \\
\hline
 \amazons & $53.604974$ &  $50.416742$\\
\hline
 \livejournals  & $> 8 $ hrs & $5592$  \\
\hline
\bonanzas & $2.075977$ & $0.26834$ (\textbf{8x}) \\
\hline
\senates & $57.722121$ & $0.481628$ (\textbf{120x})  \\
\hline
\houses & $1159.878415$ & $8.985430$ (\textbf{129x})\\
\hline

\trackers  & $> 24 $ hrs & $12471$  \\
\hline

\orkuts & $> 96 $ hrs & $232$K \\
\hline
\end{tabular}}
\end{center}
\end{table}

\remove{
\noindent
In our experiments, we utilized a total of {8} datasets to evaluate the performance of our technique. These {8} datasets included all of the {6} real datasets mentioned in our reference, which were selected to ensure fairness in our comparisons.

The 6 real datasets we used are \texttt{DBPedia} \cite{dbpedia}, \texttt{Twitter} \cite{Twitter}, \texttt{Amazon} \cite{Amazon}, \texttt{Live-journal} \cite{Live-journal}, \texttt{Delicious} \cite{Delicious}, \texttt{Tracker} \cite{Tracker}.

The summary of datasets is shown in Table \ref{table:datasets}.
}
\subsection{Performance Assessment.}
Here we evaluate the performance of the proposed algorithms on all the dataset. 

\noindent\textbf{Computation time:}~The computation time of \bucketing is substantially faster than the \baseline on most of the graphs. For example, on the \senate and \house datasets, we see that \bucketing is more than $100$ times faster than the baseline. This is because, in the \baseline, we exhaustively enumerate each butterfly and check if it is a balanced one or not, whereas, in \bucketing, we avoid computing butterflies that are not balanced by grouping symmetric and asymmetric wedges in two buckets, such that all the butterflies formed in each of the buckets is guaranteed to become balanced.

We observe that the speedup in the runtime of \textbf{\bucketing} compared to the \textbf{\baseline} is not uniform across the datasets. To understand the reason for this non-uniformity in the speedup, we studied the frequency distribution of the size of the neighborhood's typical to vertices $u$ and $w$ for all distinct vertex pairs $(u,w)$ where $u$ and $w$ are the two vertices of a butterfly located in the same side of the bipartite network. The results are shown in Figure~\ref{figure:deg-dist}. In the \baseline, we iterate over all possible pair of $2$ vertices taken from the common neighborhood of $u$ and $w$ to decide which one is balanced and which one is not. Therefore, this cost of iteration increases when the size of the common neighborhood becomes large, and this large size common neighborhood occurs more frequently, which is evident from the diagram. For example, the frequency of the common neighborhood of size more than $100$ is high in the case of \house and \senate (Figure~\ref{fig:house-deg-dst} and Figure~\ref{fig:senate-deg-dst}) compared to other three data \amazon, \dbpedia, and \bonanza (Figure~\ref{fig:amazon-deg-dst}, Figure~\ref{fig:dbpedia-deg-dst}, and Figure~\ref{fig:bonanza-deg-dst}) and therefore, the speedup of the \bucketing (compared with \baseline) for \house and \senate is much higher (more than $100\times$) than that of \amazon, \dbpedia, and \bonanza.

\begin{figure*}[t!]
	\centering
	\vspace{-3ex}
	\resizebox{1\linewidth}{!}{
	\subfloat[][{\amazon}]{\includegraphics[width=0.19\textwidth]{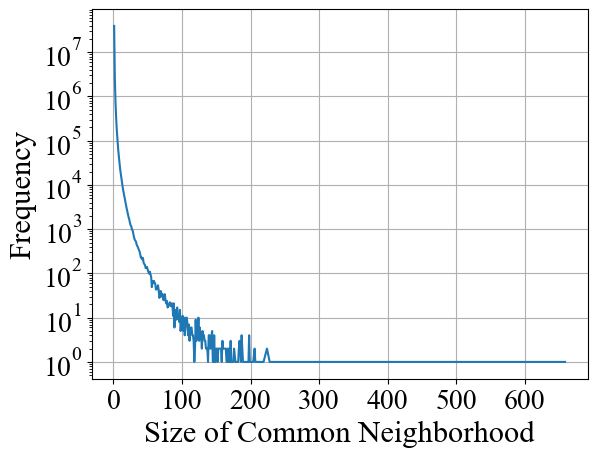}\label{fig:amazon-deg-dst}}~{}~{}
	\hspace{1ex}
	\subfloat[][{\dbpedia}]{\includegraphics[width=0.19\textwidth]{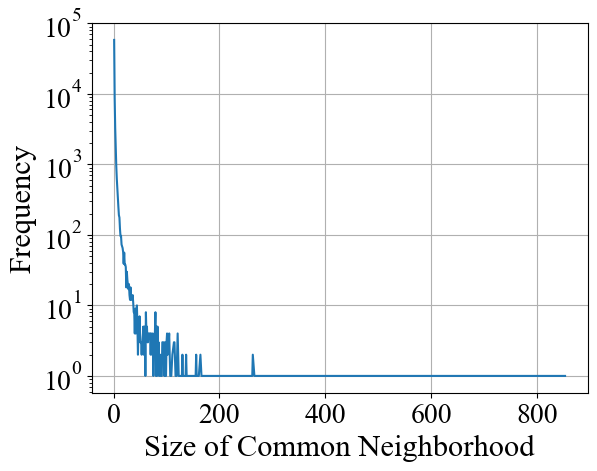}\label{fig:dbpedia-deg-dst}}~{}~{}
        \hspace{1ex}
	\subfloat[][\bonanza]{\includegraphics[width=0.19\textwidth]{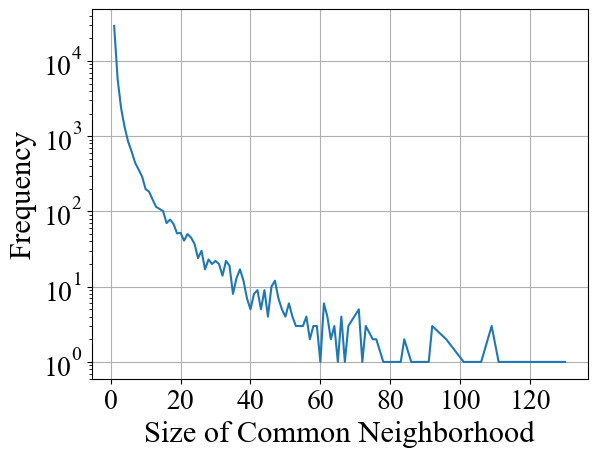}\label{fig:bonanza-deg-dst}}~{}~{}
	\hspace{1ex}
	\subfloat[][\house]{\includegraphics[width=0.19\textwidth]{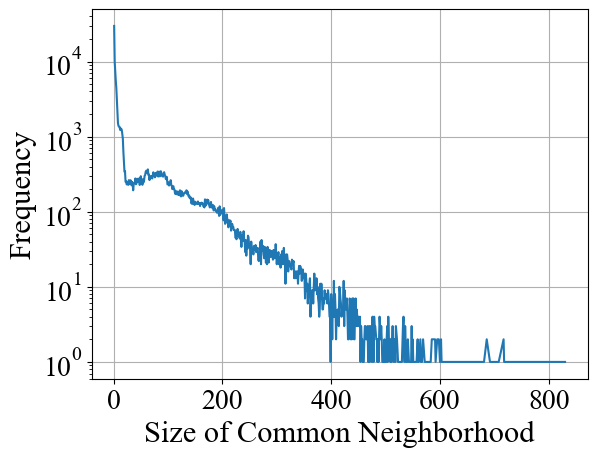}\label{fig:house-deg-dst}}~{}~{}
        \hspace{1ex}
	\subfloat[][\senate]{\includegraphics[width=0.19\textwidth]{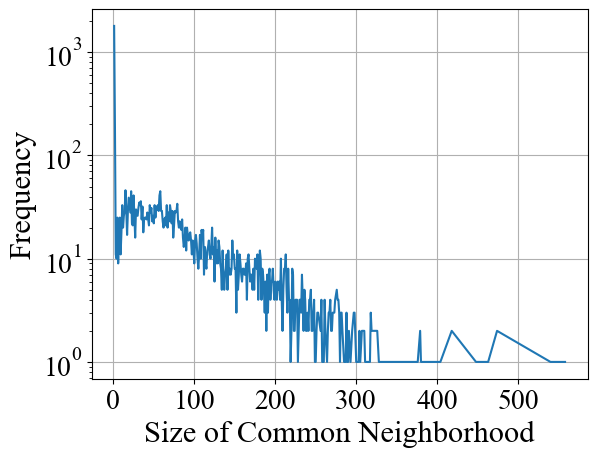}\label{fig:senate-deg-dst}}~{}~{}

	}
	\vspace{-1ex}
	\caption{ \bf {Frequency distribution of the neighborhood sizes common to $u$ and $w$ while computing the number of butterflies containing vertices $u$ and $w$ in one side of the network.}
	}\label{figure:deg-dist}
	\vspace{-3ex}
\end{figure*}


\noindent\textbf{Parallel execution time:}~In this experiment, we implemented parallel bucketing-based algorithm \parbucketing using Intel TBB and tested it on a 64-core computer. The execution time of the \parbucketing with $1$ core and $64$ cores are presented in Table~\ref{tab:parallel-time}. We observe that the parallel implementation provides decent speedup on most of the large graphs without performing any further optimizations. All the speedup results of the parallel algorithm are compared with the runtime of the parallel algorithm \parbucketing while executing in a single core. We also present the scalability result in Figure~\ref{fig:parallel-scalability} for the datasets where we obtain more than $20\times$ speedup. It shows that the speedup is almost linear with respect to the number of processors.

\remove{
\begin{figure}[h!]
\centering
\includegraphics[width=0.46\textwidth]{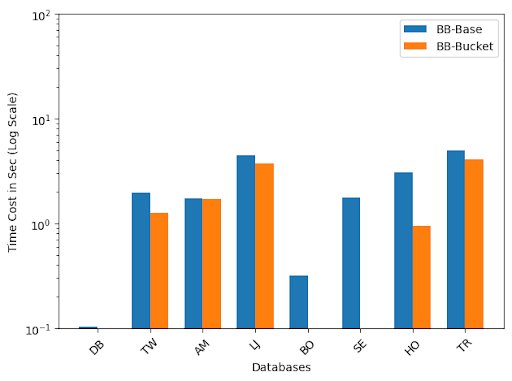}
\caption{Performance on different datasets}
\label{fig:all}
\end{figure}
}

\begin{figure}[h!]
\centering
\includegraphics[width=0.44\textwidth]{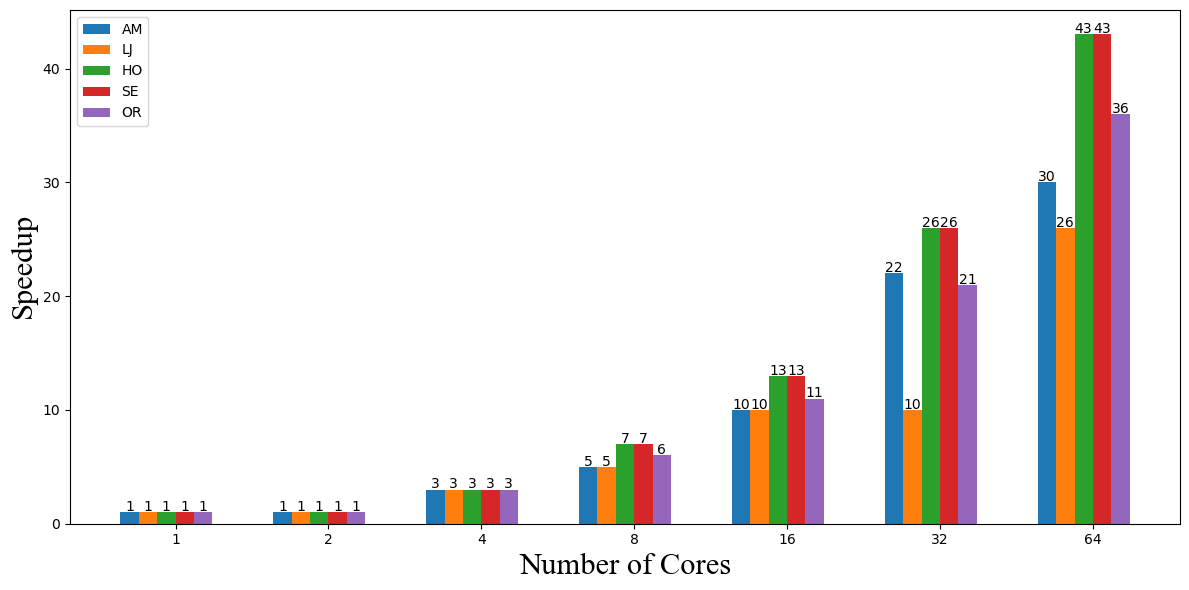}
\caption{Parallel Run time on different Cores with different Datasets}
\label{fig:parallel-scalability}
\end{figure}

\remove{
The figure {3}, displays the performance evaluation of algorithms that calculate signed Butterflies for every edge in a given graph G. The results indicate that the Bucketing approach algorithm outperforms other algorithms in terms of computational efficiency.
}

\begin{table}[t]
\small
\caption{Parallel Runtime (sec.) on real-world data sets.}
\label{tab:parallel-time}
\vspace*{-4mm}
\begin{center}
\scalebox{0.888}{
\begin{tabular}{|c|c|c|}
\hline
\textbf{Dataset}  & \parbucketing (using $1$ core) & \parbucketing (using $64$ cores)  \\
\hline
 \dbpedias & $1.39$ & $0.08$ (\textbf{17.4x})\\
\hline
 \twitters & $75.73$ & $4.63$ (\textbf{16.4x})\\
\hline
 \amazons & $160.7$ & $5.22$ (\textbf{30.8x})\\
\hline
 \livejournals  & $14503$ & $557$ (\textbf{26x})\\
\hline
\bonanzas & $0.35$ & $0.06$ (\textbf{5.8x})\\
\hline
\senates & $0.7$ & $0.03$  (\textbf{23x})\\
\hline
\houses & $10.23$ & $0.23$ (\textbf{44.5x})\\
\hline

\trackers  & $21928$ & $1747$ (\textbf{12.5x})  \\
\hline

\orkuts & $231982$ &  $6435$ (\textbf{36x})\\
\hline
\end{tabular}}
\end{center}
\vspace{-5 pt}
\end{table}
\subsection{Case Study}
We conduct a case study for counting balanced butterflies in a singed bipartite graph of \textit{movie-actor} network. 
Counting the number of balanced butterflies is able to deduce many of interesting outcomes on such data set.
In this case study we compare two list of actors: (i) top-$10$ actors with most number of positive butterflies only (Fig. \ref{fig:graphA}), say List $1$. (ii) top-$10$ actors with most number of positive edges, say List $2$.
To achieve these result, we present some modifications to the raw actor-movie data set so as to convert it in a proper signed bipartite network. The characteristics of the resulting data set used are as follows. It is composed of two distinct set of vertices, i.e., actors and movies (unique ID). Two vertices are connected by an edge if an actor as participated in a movie, and the binary edge weight $0$ or $1$ represent the positive or negative relation among the vertices. The binary sign value of $1$ is assigned to graph connections where the IMDb score for a movie was $\geq 6$, while a value of $0$ was assigned otherwise. Figure \ref{fig:Case_Study} represents a snapshot of the resulting signed bipartite graph. We now study the three questions for our case study.

The resulting List $1$ and $2$ are \{Johnny Depp, Leonard Nimoy, Nichelle Nichols, Jennifer Lawrence, Robert Downey Jr., Jason Statham, Scarlett Johansson, Orlando Bloom, Steve Buscemi
J.K. Simmons\} and \{Robert De Niro, 
Morgan Freeman, Bruce Willis, Matt Damon
Johnny Depp, Steve Buscemi, Brad Pitt
Nicolas Cage, Bill Murray, Will Ferrell\} respectively. Comparing these lists, we have following inferences: (i) \textit{Unique Actors}: List $1$ includes some unique and distinctive actors who are not present in List $2$. Actors like Leonard Nimoy, Nichelle Nichols, Jennifer Lawrence, Robert Downey Jr., Jason Statham, and Orlando Bloom bring their own unique styles and fan bases, adding diversity to List $1$.  
(ii) \textit{Iconic Characters and Cultural Impact}: List $1$ includes actors like Leonard Nimoy, Nichelle Nichols, and Johnny Depp, who have portrayed iconic characters with significant cultural impact. Leonard Nimoy's portrayal of Spock in `Star Trek', for example, has become an iconic and influential character in the science fiction genre. While List $2$ also includes Johnny Depp, List $1$ has a more diverse range of actors known for their iconic roles. 
(iii) \textit{Versatility in Genres}: List $1$ boasts actors known for their versatility across various genres. Jennifer Lawrence has showcased her talent in both dramatic and action films, while Jason Statham has established himself as an action star. This versatility allows List $1$ to cater to different genre preferences and attract a wider audience.
(iv) \textit{Overall Prominence}: Although, list $2$ includes a strong lineup of prominent actors, including Robert De Niro, Morgan Freeman, and Brad Pitt, who are widely recognized and have made significant contributions to the film industry. List $1$ has a mix of established actors like Johnny Depp and emerging talents like Jennifer Lawrence, contributing to its own level of prominence. Therefore, it helps in recognizing the upcoming talent in the movie industry.

In summary, List $1$ and List $2$ have similarities in terms of blockbuster success, recognition, and cultural impact. However, List $2$ is more sophisticated as it includes notable names and actors with a strong overall prominence. Moreover, we have also shown top-$4$ actors based on positive number of butterflies and their respective movies in Table \ref{Table:Case_study}.

\begin{figure}[t]
\centering

\includegraphics[scale=0.17]{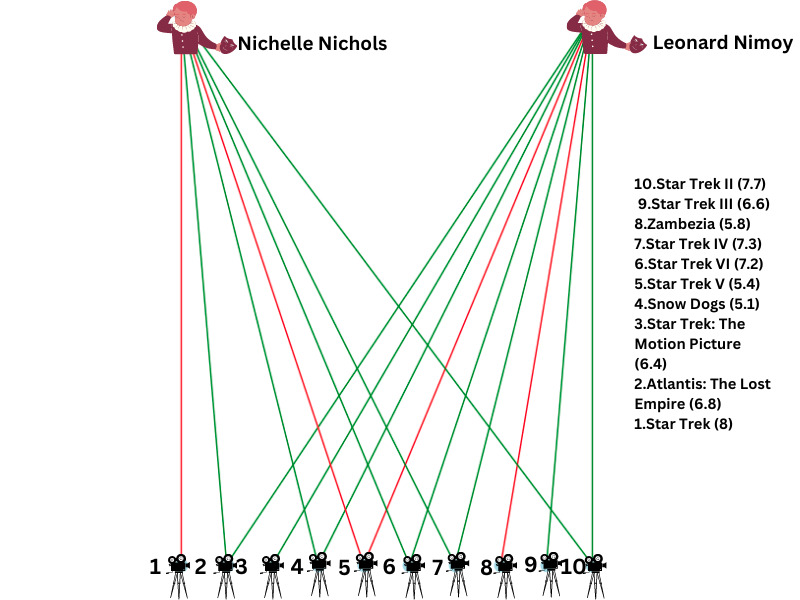}
\caption{Snapshot of the modified actor-movies data set with two actors, i.e., Nichelle Nichols and Leonard Nimoy, and their respective movies.}
\label{fig:Case_Study}

\end{figure}

\label{label:expsetting}
\noindent
\begin{table*}[ht]
\small
\caption{\label{table:datasets}
Results of Actor Movie}
\vspace*{-4mm}
\begin{center}
\resizebox{17cm}{!}{
\begin{tabular}{|c|c|c|c|}
\hline
$Actor-Johnny Depp$  & $Jennifer Lawrence$ & $Actor-Jason Flemyng
$ & $Actor-Aidan Turner
$  \\
 $15(BBF)$ & $12(BBF)$ & $7(BBF)$ & $6(BBF)$ \\
\hline
$Pirates of the Caribbean: The Curse of the Black Pearl (8.1)$ & $X-Men: Days of Future Past(8)$ & $Lock, Stock and Two Smoking Barrels
(8.2)$ & $The Hobbit: The Desolation of Smaug(7.9)$ \\
\hline
$Platoon (8.1)$ & $Star Trek VI: The Undiscovered Country (7.2)$ & $The Curious Case of Benjamin Button (7.8)$ & $The Hobbit: An Unexpected Journey(7.9)$ \\
\hline
$Edward Scissorhands (7.9)$ & $Silver Linings Playbook(7.8)$ & $Rob Roy(6.9)$ & $The Hobbit: The Battle of the Five Armies(7.5)$ \\
\hline
$Ed Wood(7.9)$ & $X-Men: First Class(7.8)$ & $From Hell(6.8)$ & $The Mortal Instruments: City of Bones(6)$ \\
\hline
$Donnie Brasco(7.8)$ & $The Hunger Games: Catching Fire(7.6)$ & $Mean Machine(6.5)$ & $ $ \\
\hline
$Finding Neverland(7.8)$ & $X-Men: Apocalypse(7.3)$ & $Transporter 2(6.3)$ & $ $ \\
\hline
$What's Eating Gilbert Grape(7.8)$ & $The Hunger Games(7.3)$ & $Deep Rising
(6)$ & $ $ \\
\hline
$Fear and Loathing in Las Vegas(7.7)$ & $American Hustle(7.3)$ & $Clash of the Titans(5.8)$ & $ $ \\
\hline
$Blow(7.6)$ & $Winter's Bone(7.2)$ & $Ironclad(6.2)$ & $ $ \\
\hline
$A Nightmare on Elm Street(7.5)$ & $The Beaver(6.7)$ & $Mirrors(6.2)$ & $ $ \\
\hline

\end{tabular}}
\end{center}
\label{Table:Case_study}
\end{table*}

\remove{
\begin{figure}
\centering
\includegraphics[width=0.46\textwidth]{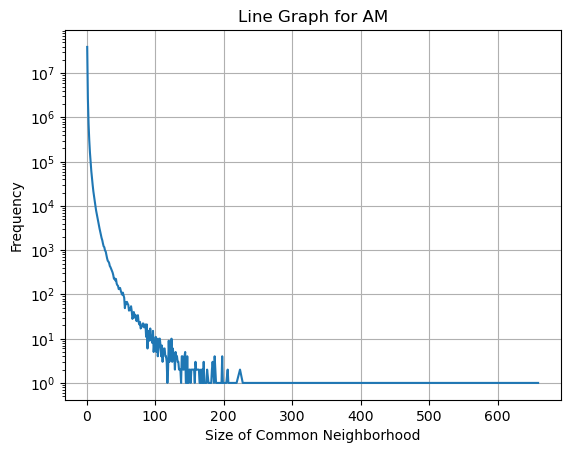}
\caption{line plot}
\label{fig:all}
\end{figure}

\begin{figure}
\centering
\includegraphics[width=0.46\textwidth]{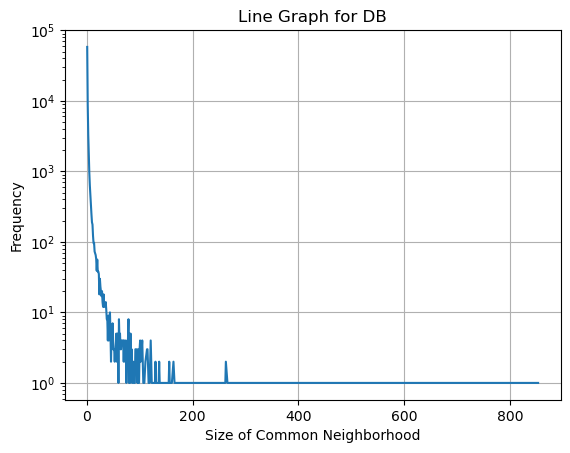}
\caption{line plot}
\label{fig:all}
\end{figure}

\begin{figure}
\centering
\includegraphics[width=0.46\textwidth]{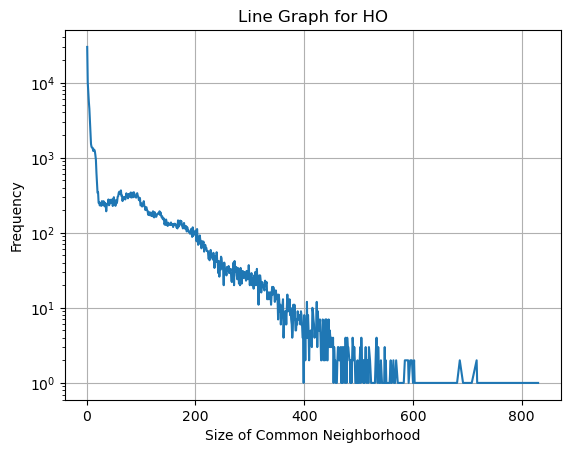}
\caption{line plot}
\label{fig:all}
\end{figure}
}

%% file: Related_Work.tex
\section{Related Work}
\noindent\textbf{Motif counting in unipartite graphs.}
Triangles are considered the smallest non-trivial cohesive structure and serve as the fundamental building block for unipartite graphs. There have been several studies on counting triangles in unipartite graphs, including \cite{stefani2017triest, al2018triangle, kolountzakis2010efficient, shun2015multicore}.
Since the butterfly is comparable to a triangle in a unipartite graph and serves as the smallest unit of cohesion in bipartite graphs \cite{abidi2022searching}, it is natural to consider leveraging the existing algorithms for triangle counting to count butterflies. However, butterfly counting in bipartite graphs presents significant challenges due to two key reasons. Firstly, the number of butterflies in a bipartite graph can be much larger than that of triangles, resulting in a higher computational complexity \cite{wang2019vertex}. Secondly, the structures of butterflies and triangles are different, with butterflies forming $4$-path cycle and triangles forming $3$-path cycle. Therefore, the existing triangle counting techniques are not directly applicable to efficiently count butterflies in bipartite graphs.

\noindent\textbf{Motifs in the bipartite network.}
Bipartite graphs possess a unique structure consisting of two independent vertex sets, which can make them challenging to analyze using traditional graph metrics. Motifs in bipartite graphs are comprised of nodes from both partitions and can reveal crucial relationships between different types of nodes in the graph. Several notable bipartite graph motifs have been proposed in the literature \cite{sariyuce2015finding, borgatti1997network, opsahl2013triadic, robins2004small, latapy2008basic}. The ($3,3$)-biclique was defined by Borgatti and Everett as the smallest motif to study cohesiveness in bipartite graphs \cite{borgatti1997network}. Opsahl introduced the closed $4$-path, which is defined as the path of four nodes contained in a cycle of at least six nodes \cite{opsahl2013triadic}. Robins and Alexander utilized the ($2,2$)-biclique to model cohesion \cite{robins2004small} and examined $3$-paths, consisting of three edges with two vertices from each set. In \cite{sariyuce2015finding}, the $k-(r,s)$-nucleus motif was proposed, where tuning the parameters $k, r,$ and $s$ resulted in motifs with hierarchical properties that address various exclusive applications. Several density measures for bipartite graphs, such as those presented in \cite{latapy2008basic}, have also been proposed. However, our focus is on subgraph motifs. Although several motifs are available for bipartite networks, the butterfly or $(2,2)$-biclique or rectangle is considered the basic building block because they are analogous to triangles in unipartite graphs \cite{aksoy2017measuring, sanei2018butterfly, wang2014rectangle}.

\noindent\textbf{Butterfly counting in bipartite graphs.}
Butterfly counting has proven to be useful for fast and global understanding of a network information, and can also be exploited in finding more complex structures such as bicliques, $k$-wing and $k$-bitrusses. A substantial body of bipartite motif's research has been devoted to the counting of butterflies, encompassing an array of sophisticated methodologies including deterministically-driven, parallelized, cache-aware, and heuristic approaches. Some of the notable works are as follows.
\cite{wang2014rectangle} was among the first algorithms to address the problem of butterfly counting in bipartite graphs. The novel vertex priority butterfly counting was proposed in \cite{wang2019vertex}, which exploited the degree of vertices on both the sides in a bipartite graph. . 
Various parallel and heuristic algorithms have also been proposed for counting butterflies in bipartite network such as parallel algorithms in \cite{xu2022efficient, weng2022distributed, shi2020parallel} and approximate approaches in \cite{li2021approximately, sanei2019fleet}. \cite{sheshbolouki2022sgrapp} studied the new problem of butterfly counting in streaming bipartite graph and proposed an approximate approach to address it.
Recently, for understanding probabilistic bipartite graphs, uncertain butterfly counting problem has been studied in \cite{zhou2021butterfly}. Some of the other recent notable motif counting in bipartite graphs also include $6$-cycle in bipartite graph \cite{niu2022counting} and Bi-triangle in bipartite graphs \cite{yang2021efficient}. 





%% file: Conclusion_and_References.tex
\section{Conclusion}

A butterfly or a $4$-cycle in a signed bipartite network is balanced with an even number of negative edges. We introduce the problem of counting balanced butterflies in a signed bipartite graph which is a fundamental problem in analyzing a signed bipartite network motivated by important applications of balanced butterflies. We first develop an efficient sequential algorithm \bucketing that can count all balanced butterflies without explicitly checking the sign of each edge. Next, we develop a parallel algorithm \parbucketing through parallelizing the iterations. We empirically evaluate real-world and synthetic datasets to show that \bucketing outperforms the baseline \baseline and the parallel algorithm \parbucketing scales almost linearly on large networks and provides a significant speedup compared to the sequential implementation.

\bibliographystyle{ACM-Reference-Format}
\bibliography{references}